\documentclass[11pt]{article}

\usepackage{amsthm, dsfont, fullpage}

\input{Qcircuit}

\date{}

\begin{document}

\title{Strong equivalence of reversible circuits is
  coNP-complete}

\author{Stephen P. Jordan \\ \small{{\em National Institute of Standards and
    Technology.} \texttt{stephen.jordan@nist.gov}}}

\bibliographystyle{unsrt}

\maketitle

\newcommand{\id}{\mathds{1}}
\newtheorem{proposition}{Proposition}
\newtheorem{definition}{Definition}
\newtheorem{corollary}{Corollary}
\newtheorem{lemma}{Lemma}

\begin{abstract}
It is well-known that deciding equivalence of logic circuits is a
coNP-complete problem. As a corollary, the problem of deciding weak
equivalence of reversible circuits, \emph{i.e.} ignoring the ancilla
bits, is also coNP-complete. The complexity of deciding strong
equivalence, including the ancilla bits, is less obvious and may depend
on gate set. Here we use Barrington's theorem to show that deciding
strong equivalence of reversible circuits built from the Fredkin gate
is coNP-complete. This implies coNP-completeness of deciding strong
equivalence for other commonly used universal reversible gate sets,
including any gate set that includes the Toffoli or Fredkin gate.
\end{abstract}

\section{Introduction}

Any Boolean circuit with $n$ bits of input and $m$ bits of output
implements a function $f:\{0,1\}^n \to \{0,1\}^m$. If two circuits
implement the same function we say they are equivalent. The following
is well-known. 

\begin{proposition}
\label{irreversible}
Deciding equivalence of Boolean circuits is coNP-complete.
\end{proposition}

Proving containment in coNP is easy; if two Boolean circuits are
inequivalent then there exists an input on which their outputs
differ, which serves as a concise, efficiently-verifiable proof of
inequivalence. The coNP-hardness follows directly from Cook's 1971
result that the logical non-tautology problem is NP-complete
\cite{Cook}. Thus, many sources cite \cite{Cook} as the origin
of proposition \ref{irreversible}.

Proposition \ref{irreversible} has interesting consequences both in
practical circuit design and in mathematical logic. On the practical
side, one may wish to reduce a given logic circuit to a normal form
dependent only on the function $f:\{0,1\}^n \to \{0,1\}^m$ that it
implements. This would achieve indistinguishability obfuscation, as
defined in \cite{BGIRSVY01, GR07}. Furthermore, deciding equivalence of
logic circuits has applications to circuit optimization, and this has
motivated the development of equivalence-checking software, which is
now included in many CAD packages \cite{Scheffer}. However,
proposition \ref{irreversible} implies no algorithm for checking
equivalence or reducing to normal form can have have polynomial
asymptotic runtime unless $\mathrm{P} = \mathrm{NP}$. (P is closed
under complement, so P=coNP implies P=NP.)

On the mathematical logic side, it is known that all equivalences
between logic propositions are generated by a small number of local
rules, such as distributivity and De Morgan's laws. (One way to
prove this is by using the rules to reduce arbitrary propositions to
disjunctive normal form \cite{BM}.) Proposition \ref{irreversible}
implies that for some pairs of equivalent propositions, the shortest
sequence of such local transformations needed to get from one to the
other must be superpolynomially long, under the standard assumption
that $\mathrm{coNP} \neq \mathrm{NP}$. (If a polynomial-length sequence
always existed, the equivalence problem would be contained in
NP. Therefore, by proposition \ref{irreversible}, $\mathrm{coNP}
\subseteq \mathrm{NP}$. $\mathrm{coNP} \subseteq \mathrm{NP}$ implies
coNP=NP, because if a language were contained in NP but not coNP then
its complement would lie in coNP but not NP.)

When considering possible generalizations of proposition
\ref{irreversible} to reversible circuits, the following two natural
definitions of equivalence present themselves.

\begin{definition}
Let $R$ be a reversible circuit on $b$ bits. By initializing the last
$b - n$ input bits to zero and ignoring the last $b - m$ output bits, 
$(R,n,m)$ defines a Boolean function $f_{R,n,m}:\{0,1\}^n \to \{0,1\}^m$.
$(R,n,m)$ is weakly equivalent to $(R',n,m)$ if $f_{R,n,m} =
f_{R',n,m}$.
\end{definition}

\begin{definition}
Let $R$ be a reversible circuit on $b$ bits. $R$ defines a bijection
$f_R:\{0,1\}^b \to \{0,1\}^b$. $R$ is strongly equivalent to $R'$ if
$f_R = f_{R'}$.
\end{definition}

Weak equivalence of reversible circuits is easily seen to be a
coNP-complete problem. This follows from proposition
\ref{irreversible} and the computational universality of reversible
circuits, which was proven in \cite{Fredkin_Toffoli}, building upon
\cite{Bennett}.

It is also clear that the problem of deciding strong
equivalence of reversible circuits is contained in coNP. However, the
question whether strong equivalence of reversible circuits is a
coNP-hard problem is more subtle and may depend on gate-set. Two of
the most standard reversible gates are the Fredkin gate and the
Toffoli gate, described in Figure \ref{FredTof}. The Fredkin gate is
computationally universal by itself, as is the Toffoli gate
\cite{Fredkin_Toffoli}. Our main result is the following.

\begin{proposition}
\label{Fredkinprop}
The problem of deciding strong equivalence between reversible circuits
constructed from the Fredkin gate is coNP-hard.
\end{proposition}

\begin{figure}
\[
\begin{array}{cc|cc}
\textrm{Fredkin} & \quad & \quad & \textrm{Toffoli} \\
\hline
 & & & \\
\begin{array}{cc}
\begin{tabular}{c|c}
\textrm{in} & \textrm{out} \\
\hline
000 & 000 \\
001 & 001 \\
010 & 010 \\
011 & 011 \\
100 & 100 \\
101 & 110 \\
110 & 101 \\
111 & 111 \\
\end{tabular}
&
\Qcircuit @C=1em @R=1.5em {
      & \ctrl{1}                     & \qw \\
      & \multigate{1}{\mathrm{SWAP}} & \qw \\
      & \ghost{\mathrm{SWAP}}        & \qw
}
\end{array}
& \quad & \quad &
\begin{array}{cc}
\begin{tabular}{c|c}
\textrm{in} & \textrm{out} \\
\hline
000 & 000 \\
001 & 001 \\
010 & 010 \\
011 & 011 \\
100 & 100 \\
101 & 101 \\
110 & 111 \\
111 & 110
\end{tabular}
&
\Qcircuit @C=1em @R=1.5em {
      & \ctrl{1} & \qw \\
      & \ctrl{1} & \qw \\
      & \targ    & \qw
}
\end{array}
\end{array}
\]
\caption{\label{FredTof} The Fredkin gate swaps the second two bits if
  the first bit is 1. The Toffoli gates flips the last bit if the
  first two bits are both 1. Shown here are the truth tables and
  corresponding circuit diagrams.} 
\end{figure}

\noindent
A Fredkin gate can be constructed from three Toffoli gates. Thus,
proposition \ref{Fredkinprop} immediately yields the following.

\begin{corollary}
\label{Toffoliprop}
The problem of deciding strong equivalence between reversible circuits
constructed from the Toffoli gate is coNP-hard.
\end{corollary}

Proposition \ref{Fredkinprop} and corollary \ref{Toffoliprop} can be
viewed as classical analogues to the quantum hardness results
regarding the non-identity problem for quantum circuits \cite{JWB05,
  T10}. Furthermore, our proof uses techniques related to those
used in \cite{JWB05, T10}. 

Just as for conventional irreversible circuits, software packages have
been developed for checking equivalence of reversible circuits,
motivated by applications to circuit optimization 
\cite{Wille, Yamashita}. Proposition \ref{Fredkinprop} implies that for
standard reversible gate sets, such software cannot have polynomial
asymptotic runtime assuming $\mathrm{P} \neq \mathrm{coNP}$
(equivalently, $\mathrm{P} \neq \mathrm{NP}$). Furthermore, all strong
equivalences between Toffoli circuits are generated by repeated
application of a finite set of local equivalence rules
\cite{Iwama}. Thus, corollary \ref{Toffoliprop} implies that for some
pairs of strongly equivalent Toffoli circuits, the number of
applications of the local equivalence rules to get from one to the
other must be superpolynomially long, assuming $\mathrm{NP} \neq
\mathrm{coNP}$.

\section{Proof}

We start by reviewing our two main tools: Cook's theorem and
Barrington's theorem. Recall that a clause is a set of literals and
negated literals joined by OR, and a Conjunctive Normal Form (CNF)
formula is a set of clauses joined by AND. Cook's theorem states that
the problem of deciding satisfiability of CNF formulas is NP-complete
\cite{Cook}. As a corollary of Cook's theorem one has the following.

\begin{corollary}
\label{3unsat}
The problem of deciding unsatisfiability of CNF formulas is
coNP-complete.
\end{corollary}

\noindent
Barrington's theorem quantifies the power of a highly space-limited
model of computation called width-5 branching programs. Barrington's
proof of this theorem relies on the fact that $S_5$ is a non-solvable
group \cite{Barrington}.

\begin{definition}
A length $l$ width-5 branching program taking $n$ bits of input is a
sequence of $l$ triples, each of the form $(i,\alpha,\beta)$, where
$i \in \{1,\ldots,n\}$, and $\alpha, \beta$ are permutations
from $S_5$. The triple is interpreted as an instruction to apply
permutation $\alpha$ if the $i^{\mathrm{th}}$ input bit is zero,
or apply $\beta$ if the $i^{\mathrm{th}}$ input bit is one. The
final permutation obtained by the composition of the $l$ permutations
is the output of the branching program.
\end{definition}

\begin{proposition}{(Barrington's theorem \cite{Barrington})}
\label{Bthm}
Given any depth $d$, fan-in 2 Boolean formula $f$, and any 5-cycle $\alpha \in
S_5$, one can in $\mathrm{poly(4^d)}$ time construct a width-5
branching program of length at most $4^d$ such that the branching
program evaluates to $\alpha$ if $f$ is TRUE and to $\id$
otherwise. 
\end{proposition}

\noindent
To apply Barrington's theorem toward proving proposition
\ref{Fredkinprop}, we first prove the following lemma.

\begin{lemma}
\label{partway}
Let $P$ be a length-$l$ width-5 branching program on $n$ input
bits. Given $P$, one can construct a circuit of $O(l)$ Fredkin gates
acting on $n + 6$ bits that permutes five ancilla bits according to
the output of $P$, provided the sixth ancilla bit is initialized to
one. Furthermore, the value of the ancilla bit is always left
unmodified by this circuit.
\end{lemma}

\begin{proof}
The transpositions (``swaps'') generate $S_5$ (or any symmetric
group). Thus, for any triple of the form $(i, \id, \beta)$ we can
construct a  corresponding sequence of $O(1)$ Fredkin gates controlled
by $i$ such that the ancillary bits are permuted according to
$\beta$ if the $i^{\mathrm{th}}$ input bit is one and are left
untouched otherwise. We can use the ancilla initialized to one as the
control bit of a Fredkin gate, thereby simulating a SWAP gate. A
Fredkin gate followed by a SWAP gate on the target bits swaps the
target bits if and only if the control bit is zero. This in turn
allows us to implement triples of the form $(i,\alpha,\id)$. By
composing the triples $(i, \id,\beta)$ and $(i, \alpha, \id)$ one then
obtains an arbitrary triple $(i, \alpha, \beta)$. By simulating the
full sequence of $l$ such triples, one simulates the full branching
program. Because each Fredkin gate involving the ancilla bit uses the
ancilla only as a control bit, its value is left untouched for all
possible inputs.
\end{proof}

\noindent
We now prove our main result.

\begin{proof}{(of Proposition \ref{Fredkinprop})}
Polynomial size CNF formulas can always be expressed as logarithmic
depth circuits. Thus, by proposition \ref{Bthm} and corollary
\ref{3unsat}, the problem of deciding whether a given width-5 branching
program always evaluates to the identity is coNP-hard. Lemma
\ref{partway} gets us part of the way toward using this fact to prove
that deciding whether a Fredkin circuit is equivalent to the identity
is coNP-hard. However, lemma \ref{partway} assumes the presence of an
ancilla bit initialized to one, whereas strong equivalence means
equivalence on all inputs. We can work around this problem by
simulating the presence of an ancilla bit initialized to one using the
following reversible circuit. 
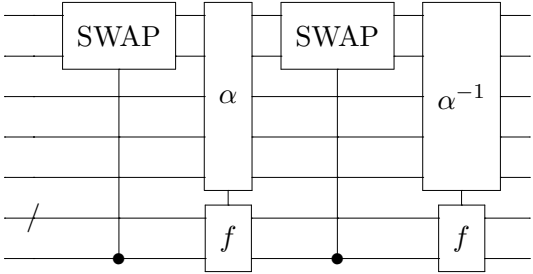
\vspace{11pt}
\[
\Qcircuit @C=1em @R=.5em {
 & \qw     & \multigate{1}{\mathrm{SWAP}} & \multigate{4}{\alpha}  & \multigate{1}{\mathrm{SWAP}} & \multigate{4}{\alpha^{-1}}  & \qw \\
 & \qw     & \ghost{\mathrm{SWAP}}        & \ghost{\alpha}         & \ghost{\mathrm{SWAP}}        & \ghost{\alpha^{-1}}         & \qw \\
 & \qw     & \qw                          & \ghost{\alpha}         & \qw                          & \ghost{\alpha^{-1}}         & \qw \\
 & \qw     & \qw                          & \ghost{\alpha}         & \qw                          & \ghost{\alpha^{-1}}         & \qw \\
 & \qw     & \qw                          & \ghost{\alpha} \qwx[1] & \qw                          & \ghost{\alpha^{-1}} \qwx[1] & \qw \\
 & {/} \qw & \qw                          & \multigate{1}{f}       & \qw                          & \multigate{1}{f}           & \qw \\
 & \qw     & \ctrl{-5}                    & \ghost{f}              & \ctrl{-5}                    & \ghost{f}                  & \qw
}
\]
Here the slash is a shorthand for $n$ bits. The first (leftmost)
Fredkin gate swaps the top two bits if the bottom bit is initialized
to one. Next, the linked boxes labeled $\alpha$ and $f$ are a
shorthand for the circuit constructed in lemma \ref{partway}. That is,
they are a sequence of Fredkin gates, which, under the assumption
that the bottom bit is 1, apply the 5-cycle $\alpha = (12345)$ to bits
one through five if the CNF formula $f$ evaluates TRUE, and the
identity permutation otherwise. This action is followed by a second
Fredkin gate swapping the top two bits controlled by the bottom bit,
and lastly the inverse of the controlled-$\alpha$ circuit. (Because
Fredkin gates are self-inverse, one obtains this inverse simply by
reversing the order of the gates.)

We now analyze case-by-case to show that the circuit described above
is strongly equivalent to the identity if and only if $f$ is
unsatisfiable. If the bottom bit is initialized to one and $f$
evaluates to true, then the top five bits are permuted nontrivially
because $(12)(12345)(12)(12345)^{-1} \neq \id$. If the bottom bit is 
initialized to one and $f$ evaluates to false, then the circuit acts
as the identity, because the two SWAPs of the top two bits cancel.

Lastly, suppose the bottom bit is initialized to zero. In this case,
the controlled-$\alpha$ circuit behaves in a way that we have not
explicitly described. However, we do know from lemma \ref{partway}
that the value of the bottom bit is left unmodified by the
controlled-$\alpha$ circuit. This is all we need to know; it 
implies that neither of the two controlled-SWAPs act, and therefore,
the (unspecified) action of the controlled-$\alpha$ circuit is
cancelled by the action of its inverse circuit.
\end{proof}

\noindent
\textbf{Acknowledgments:} I thank Scott Aaronson, Gorjan Alagic,
Stacey Jeffery, Vincent Liew, and Yi-Kai Liu for helpful
discussions. This paper is a contribution of the National Institute of
Standards and Technology and is not subject to U.S. copyright.

\bibliography{conp}

\end{document}